\documentclass[reqno,11pt]{amsart}
\usepackage{amsmath,amssymb,dsfont}
\newtheorem{theorem}{Theorem}[section]

\newtheorem{lemma}[theorem]{Lemma}
\newtheorem{proposition}[theorem]{Proposition}
\theoremstyle{definition}
\newtheorem{definition}[theorem]{Definition}
\theoremstyle{remark} \theoremstyle{remark}

\newtheorem{remark}[theorem]{Remark}
\numberwithin{equation}{section}
\DeclareMathOperator*{\esssup}{\mathrm{ess\,sup}}

\allowdisplaybreaks[3]

\title[Continuum Jump Model with Attraction]{Evolution of States of a Continuum Jump Model with Attraction}

\author{ Yuri  Kozitsky}

\address{Instytut Matematyki, Uniwersytet Marii Curie-Sk{\l}odowskiej, 20-031 Lublin, Poland}
\email{jkozi@hektor.umcs.lublin.pl}

\begin{document}



\begin{abstract}

We study a model of an infinite system of point particles in
$\mathds{R}^d$ performing random jumps with attraction. The system's
states are probability measures on the space of particle
configurations, and their evolution is described by means of
Kolmogorov and Fokker-Planck equations. Instead of solving these
equations directly we deal with correlation functions evolving
according to a hierarchical chain of differential equations, derived
from the Kolmogorov equation. Under quite natural conditions imposed
on the jump kernels -- and analyzed in the paper -- we prove that
this chain has a unique classical  sub-Poissonian solution on a
bounded time interval. This gives a partial answer to the question
whether the sub-Poissonicity is consistent with any kind of
attraction. We also discuss possibilities to get a complete answer
to this question.

\end{abstract}

\maketitle


\section{Introduction}

\label{S1}

\subsection{Setup}
In this work, we deal with the  model introduced and studied in
\cite{BKKut}. It describes an infinite system of point particles
placed in $\mathds{R}^d$ which perform random jumps with attraction.
To the best of our knowledge, \cite{BKKut} and the present research
are the only works where the dynamics of an infinite particle system
of this kind has been studied hitherto.

The phase space of the model is the set $\Gamma$ of all subsets
$\gamma \subset \mathds{R}^d$ such that the set $\gamma\cap\Lambda$
is finite whenever $\Lambda \subset \mathds{R}^d$ is compact. It is
equipped with a topology, see below, and thus with a $\sigma$-field
of measurable subsets. Thereby, one can consider probability
measures on $\Gamma$ as states of the system. Among them there are
Poissonian states in which the particles are independently
distributed over $\mathds{R}^d$. Such states are completely
characterized by the density of the particles. In
\emph{sub-Poissonian} states,  the dependence between the  positions
of the particles is controlled in a certain way (see the next
subsection),  and the particles' density is still an important
characteristic of the state. For an infinite particle system with
repulsion, in \cite{BKKK} the evolution of the system's states
$\mu_0 \mapsto \mu_t$ in the set of sub-Poissonian measures was
shown to hold  for $t< T$ with some $T < \infty$. Then in \cite{BK1}
this result was improved by constructing the global in time
evolution of states. Thus, a paramount question regarding such
models is whether the sub-Poissonicity is consistent with some sort
of attraction, and -- if yes -- for which sort and on which time
intervals. In this paper, we give a partial answer to this question.
Namely, we present quite a reasonable condition on the attraction,
see (\ref{As4}) below, under which -- as we show -- the correlation
functions evolve $k_0 \mapsto k_t$ and remain sub-Poissonian on a
bounded time interval. This result extends the corresponding result
of \cite{BKKut} in the following directions: (i) the  evolution $k_0
\mapsto k_t$ is constructed as a classical solution of the
corresponding Cauchy problem, not in a weak sense; (ii) our result
is valid for much more general types of attraction (see subsections
3.3 and 3.4 below). At the same time, the following problems remain
open: (a) proving that each $k_t$ is the correlation function of a
unique sub-Poissonian state; (b) continuing the evolution $k_0
\mapsto k_t$ to all $t>0$. In subsection 3.4 below, we discuss
possibilities to solve them.

\subsection{Presenting the result}
States of an infinite particle system are usually characterized by
means of their values $\mu(F)$ on {\it observables} $F:\Gamma \to
\mathds{R}$, defined  as
\begin{equation*}
  \mu(F) = \int_{\Gamma} F d \mu.
\end{equation*}
The system's evolution is supposed to be Markovian and hence
described by the Kolmogorov equation
\begin{equation}
  \label{1}
 \frac{d}{dt} F_t = L F_t, \qquad F_t|_{t=0} = F_0,
\end{equation}
where the operator $L$ specifies the model. Alternatively, the
evolution of states is derived from the Fokker-Planck equation
\begin{equation}
  \label{1a}
\frac{d}{dt} \mu_t = L^* \mu_t, \qquad \mu_t|_{t=0} = \mu_0,
\end{equation}
related to that in (\ref{1}) by the duality $\mu_t(F_0) =
\mu_0(F_t)$. For the model considered in this work, the operator $L$
is
\begin{eqnarray}
  \label{1L}
 (LF)(\gamma)  = \sum_{x\in \gamma} \int_{\mathds{R}^d}
 a(x,y)\left[1+\epsilon (x,y|\gamma)\right]\left[ F(\gamma\setminus x \cup y) -
 F(\gamma)\right] d y ,
\end{eqnarray}
with
\begin{equation}
  \label{M1A}
\epsilon (x,y|\gamma) = \sum_{z\in \gamma\setminus x} b(x,y|z).
\end{equation}
The quantity $b(x,y|z)\geq 0$ describes the increase of the jump
rate from $x\in \gamma$ to $y\in \mathds{R}^d$ caused by the
particle located at $z\in \gamma\setminus x$. Then $\epsilon
(x,y|\gamma)$ is the (multiplicative) increase of the corresponding
jump rate caused by the whole configuration $\gamma$. For $\epsilon
\equiv 0$, (\ref{1L}) turns into the generator of free jumps, see,
e.g., \cite{BK}.

As is usual for models of this kind, the direct meaning of (\ref{1})
or (\ref{1a}) can only be given for states of finite systems, cf.
\cite{K}. In this case, the Banach space where the Cauchy problem in
(\ref{1a}) is defined can be the space of signed measures with
finite variation. For infinite systems, the evolution is described
by means of correlation functions, see \cite{BKKK,BKKut,KK} and the
references quoted in these works. In the present paper, we follow
this approach the main idea of which can be outlined as follows. Let
$\varTheta$ be the set  of all compactly supported continuous
functions $\theta:\mathbb{R}^d\to (-1,0]$. For a state $\mu$, its
{\it Bogoliubov} functional $B_\mu : \varTheta \to \mathds{R}$ is
set to be
\begin{equation}
  \label{I1}
B_\mu (\theta) = \int_{\Gamma} \prod_{x\in \gamma} ( 1 + \theta (x))
\mu( d \gamma), \qquad \theta \in \varTheta.
\end{equation}
The function $\gamma \mapsto \prod_{x\in \gamma} ( 1 + \theta (x))$
is bounded and measurable for each  $\theta \in \varTheta$; hence,
(\ref{I1}) makes sense for each measure.
 For the homogeneous Poisson measure $\pi_\varkappa$,
$\varkappa>0$, we have
\begin{equation*}
B_{\pi_\varkappa} (\theta) = \exp\left(\varkappa
\int_{\mathbb{R}^d}\theta (x) d x \right).
\end{equation*}
In state $\pi_\varkappa$, the particles are independently
distributed over $\mathds{R}^d$ with density $\varkappa$. The set of
{\it sub-Poissonian} states $\mathcal{P}_{\rm exp}(\Gamma)$ is then
defined as that containing all those states $\mu$ for which $B_\mu$
can be continued to an exponential type entire function of $\theta
\in L^1 (\mathbb{R}^d)$. This means that it can be written down in
the form
\begin{eqnarray}
  \label{I3}
B_\mu(\theta) = 1+ \sum_{n=1}^\infty
\frac{1}{n!}\int_{(\mathbb{R}^d)^n} k_\mu^{(n)} (x_1 , \dots , x_n)
\theta (x_1) \cdots \theta (x_n) d x_1 \cdots d x_n,
\end{eqnarray}
where $k_\mu^{(n)}$ is the $n$-th order correlation function of the
state $\mu$. It is a symmetric element of $L^\infty
((\mathbb{R}^d)^n)$ for which
\begin{equation}
\label{I4}
  \|k^{(n)}_\mu \|_{L^\infty
((\mathbb{R}^d)^n)} \leq C \exp( \vartheta n), \qquad n\in
\mathbb{N}_0,
\end{equation}
with some $C>0$ and $\vartheta \in \mathbb{R}$. Sometimes,
(\ref{I4}) is called \emph{Ruelle bound}, cf. \cite[Chapter
4]{Ruelle}. Note that (\ref{I3}) can be viewed as an analog of the
Taylor expansion of the characteristic function of a probability
measure. That is why, $k^{(n)}_\mu$ are also called \emph{moment
functions}. Their evolution is described by a chain of differential
equations derived from that in (\ref{1}). The central problem of
this work is the existence of classical solutions of this chain
satisfying (\ref{I4}) with possibly time-dependent $C$ and
$\vartheta$. Its solution is given in Theorem \ref{1tm}, formulated
in subsection 3.2 and proved in Section \ref{S4}. In Section
\ref{S2}, we give some necessary information on the methods used in
the paper and specify the model. In subsection 3.1, we place the
mentioned chain of equations into suitable Banach spaces, that is
mostly performed by defining the corresponding operators. Then we
formulate Theorem \ref{1tm} and analyze the assumptions regarding
the jump kernels under which we then prove this statement. In
subsection 3.4, we give some comments on the result and the
assumptions, including discussing open problems related to the
model, and compare our result with the corresponding result of
\cite{BKKut}. Section \ref{S4} is dedicated to the proof of Theorem
\ref{1tm}.

\section{Preliminaries and the Model}

\label{S2}

Here we briefly recall the main notions relevant to the subject --
for further information we refer to \cite{Albev,BKKK,BKKut,KK} and
the literature quoted in these works.

\subsection{Configuration spaces}

Let $\mathcal{B}(\mathds{R}^d)$ and $\mathcal{B}_{\rm
b}(\mathds{R}^d)$ denote the sets of all Borel and all bounded Borel
subsets of $\mathds{R}^d$, respectively. The configuration space
$\Gamma$, equipped with the vague topology, is homeomorphic to a
separable metric (Polish) space, cf. \cite{Albev,Tobi}. Let
$\mathcal{B}(\Gamma)$ be the corresponding Borel $\sigma$-field. For
$\Lambda\in \mathcal{B}(\mathds{R}^d)$, the set $\Gamma_\Lambda =
\{\gamma\in \Gamma: \gamma \subset \Lambda\}$ is clearly in
$\mathcal{B}(\Gamma)$, and hence
\[
 \mathcal{B}(\Gamma_\Lambda):=\{ A \cap \Gamma_\Lambda : A \in \mathcal{B}(\Gamma)\}
\]
is a sub-field of $\mathcal{B}(\Gamma)$. The projection
$p_{\Lambda}:\Gamma\to \Gamma_\Lambda$ defined by  $p_\Lambda
(\gamma) = \gamma_\Lambda=\gamma \cap \Lambda$ is  measurable. Then,
for each Borel $\Lambda$ and $A_\Lambda \in
\mathcal{B}(\Gamma_\Lambda)$, we have that
\begin{equation*}
 p^{-1}_\Lambda(A_\Lambda) :=\{ \gamma\in \Gamma: p_\Lambda (\gamma) \in A_\Lambda \} \in
 \mathcal{B}(\Gamma).
\end{equation*}
Let $\mathcal{P}(\Gamma)$ denote the set of all  probability
measures on $(\Gamma, \mathcal{B}(\Gamma))$. For a given $\mu\in
\mathcal{P}(\Gamma)$, its projection on $(\Gamma_\Lambda,
\mathcal{B} (\Gamma_\Lambda))$ is defined as
\begin{equation}
 \label{5}
\mu^\Lambda (A_\Lambda) = \mu\left(p^{-1}_\Lambda (A_\Lambda)
\right), \qquad A_\Lambda \in \mathcal{B}(\Gamma_\Lambda).
\end{equation}
Let $\Gamma_0$ be the set of all finite $\gamma \in \Gamma$. Then
$\Gamma_0\in \mathcal{B}(\Gamma)$ as each of $\gamma \in \Gamma_0$
lies in some $\Lambda \in \mathcal{B}_{\rm b}(\mathds{R}^d)$, and
hence belongs to $\Gamma_\Lambda$. It can be proved that a function
$G:\Gamma_0 \to \mathds{R}$ is
$\mathcal{B}(\Gamma)/\mathcal{B}(\mathds{R} )$-measurable if and
only if, for each $n\in \mathds{N}_0$, there exists a symmetric
Borel function $G^{(n)}: (\mathds{R}^{d})^{n} \to \mathds{R}$ such
that
\begin{equation}
 \label{7}
 G(\eta) = G^{(n)} ( x_1, \dots , x_{n}),
\end{equation}
for $\eta = \{ x_1, \dots , x_{n}\}$ .
\begin{definition}
  \label{Gdef}
A measurable function $G:\Gamma_0 \to \mathds{R}$ is said have
bounded support if: (a) there exists $\Lambda \in \mathcal{B}_{\rm
b} (\mathds{R}^d)$ such that $G(\eta) = 0$ whenever $\eta\cap
\Lambda^c \neq \emptyset$; (b) there exists $N\in \mathds{N}_0$ such
that $G(\eta)=0$ whenever $|\eta|
>N$. Here $\Lambda^c := \mathds{R}^d
\setminus \Lambda$ and $|\cdot |$ stands for cardinality.
\end{definition}
The Lebesgue-Poisson measure $\lambda$ on $(\Gamma_0,
\mathcal{B}(\Gamma_0))$ is defined by the following formula
\begin{eqnarray}
\label{8} \int_{\Gamma_0} G(\eta ) \lambda ( d \eta)  = G(\emptyset)
+ \sum_{n=1}^\infty \frac{1}{n! } \int_{(\mathds{R}^d)^{n}} G^{(n)}
( x_1, \dots , x_{n} ) d x_1 \cdots dx_{n},
\end{eqnarray}
which has to hold for all $G\in B_{\rm bs}(\Gamma_0)$.

In this work, we use the following (real) Banach spaces of functions
$g: \Gamma_0 \to \mathds{R}$. The first group consists of the spaces
$\mathcal{G}_\vartheta = L^1 (\Gamma_0 , w_\vartheta d \lambda)$,
indexed by  $\vartheta \in \mathds{R}$. Here we have set
$w_\vartheta (\eta) = \exp\left( \vartheta |\eta| \right)$. Hence
the norm of $\mathcal{G}_\vartheta$ is
\begin{equation}
  \label{w1}
 |g|_{\vartheta} = \int_{\Gamma_0} |g(\eta)| w_\vartheta (\eta)
 \lambda (d \eta).
\end{equation}
Along with this norm we also consider
\begin{equation}
  \label{w2}
 \| g \|_\vartheta := \esssup_{\eta \in \Gamma_0}\left\{ |g (\eta)| \exp\big{(} - \vartheta
  |\eta| \big{)} \right\},
\end{equation}
and then set $\mathcal{K}_\vartheta =\{ g:\Gamma_0 \to \mathds{R}:
\|g\|_\vartheta < \infty\}$. These spaces constitute the second
group which we use in the sequel. From (\ref{w1}) and (\ref{w2}) we
see that $\mathcal{K}_\vartheta$ is the dual space to
$\mathcal{G}_\vartheta$ with the duality
\begin{equation}
  \label{9f}
(G, k) \mapsto \langle \! \langle G, k \rangle \!\rangle :=
\int_{\Gamma_0} G(\eta) k(\eta) \lambda (d \eta),
\end{equation}
holding for $G\in\mathcal{G}_\vartheta$ and
$k\in\mathcal{K}_\vartheta$. Note that $B_{\rm bs}(\Gamma_0)$ is
contained in each $\mathcal{G}_\vartheta$ and each
$\mathcal{K}_\vartheta$, $\vartheta \in \mathds{R}$.

For $G\in B_{\rm bs}(\Gamma)$, we set
\begin{equation}
  \label{9a}
(KG)(\gamma) = \sum_{\eta \subset \gamma} G(\eta),
\end{equation}
where the sum is taken over all finite $\eta$.

\subsection{Correlation functions}
For a given $\mu \in \mathcal{P}_{\rm exp}(\Gamma)$, similarly as in
(\ref{7}) we introduce $k_\mu : \Gamma_0 \to \mathds{R}$ such that
$k_\mu(\emptyset) = 1$ and $k_\mu(\eta) = k^{(n)}_\mu (x_1, \dots ,
x_n)$ for $\eta = \{x_1, \dots , x_n\}$, $n\in \mathds{N}$, cf.
(\ref{I1}) and (\ref{I3}). With the help of the measure introduced
in (\ref{8}), the formulas  in (\ref{I1}) and (\ref{I3}) can be
combined into the following
\begin{eqnarray*}
 B_\mu (\theta)& = & \int_{\Gamma_0} k_\mu(\eta) \prod_{x\in \eta} \theta (x) \lambda (d\eta)=: \int_{\Gamma_0} k_\mu(\eta) e( \eta; \theta) \lambda (d \eta)
 \\[.2cm]
 & = &  \int_{\Gamma} \prod_{x\in \gamma} (1+ \theta (x)) \mu (d \gamma) =: \int_{\Gamma} F_\theta (\gamma) \mu(d
 \gamma). \nonumber
\end{eqnarray*}
Thereby, we can transform  the action of $L$ on $F$, as in
(\ref{1L}), to the action of $L^\Delta$ on $k_\mu$ according to the
rule
\begin{equation}
  \label{1g}
\int_{\Gamma}(L F_\theta) (\gamma) \mu(d \gamma) = \int_{\Gamma_0}
(L^\Delta k_\mu) (\eta) e(\eta;\theta)
 \lambda (d \eta).
\end{equation}
This will allow us to pass from (\ref{1}) to the corresponding
Cauchy problem for the correlation functions, cf. (\ref{16}) below.
The main advantage here is that $k_\mu$ is a function of {\em
finite} configurations.

For $\mu \in \mathcal{P}_{\rm exp}(\Gamma)$ and  $\Lambda \in
\mathcal{B}_{\rm b}(\mathds{R}^d)$, let $\mu^\Lambda$ be as in
(\ref{5}). Then $\mu^\Lambda$ is absolutely continuous with respect
to the restriction  $\lambda^\Lambda$ to
$\mathcal{B}(\Gamma_\Lambda)$  of the measure defined in (\ref{8}),
and hence we may write
\begin{equation}
\label{9c} \mu^\Lambda (d \eta ) = R^\Lambda_\mu (\eta)
\lambda^\Lambda ( d \eta), \qquad \eta \in \Gamma_\Lambda.
\end{equation}
Then the correlation function $k_\mu$ and the Radon-Nikodym
derivative $R_\mu^\Lambda$ satisfy
\begin{eqnarray*}
k_\mu(\eta) & = & \int_{\Gamma_\Lambda} R^\Lambda_\mu (\eta \cup
\xi) \lambda^\Lambda ( d\xi).
\end{eqnarray*}
By (\ref{9a}), (\ref{5}), and (\ref{9c}) we get
\begin{equation*}
\int_{\Gamma} (KG)(\gamma) \mu(d\gamma) = \langle \! \langle G,
k_\mu \rangle \!\rangle,
\end{equation*}
holding for each $G\in B_{\rm bs}(\Gamma_0)$ and $\mu \in
\mathcal{P}_{\rm exp}(\Gamma)$, see (\ref{9f}). Define
\begin{equation*}
B^\star_{\rm bs} (\Gamma_0) =\{ G\in B_{\rm bs}(\Gamma_0):
(KG)(\gamma) \geq 0 \ {\rm for} \ {\rm all} \ \gamma\in \Gamma\}.
\end{equation*}
By \cite[Theorems 6.1 and 6.2 and Remark 6.3]{Tobi} one can prove
the next statement.
\begin{proposition}
  \label{Gpn}
Let  a measurable function $k : \Gamma_0 \to \mathds{R}$  have the
following properties:
\begin{eqnarray*}
& (a) & \ \langle \! \langle G, k \rangle \!\rangle \geq 0, \qquad
{\rm for} \ {\rm all} \ G\in B^\star_{\rm bs} (\Gamma_0);\\[.2cm]
& (b) & \ k(\emptyset) = 1; \qquad (c) \ \ k(\eta) \leq
 C^{|\eta|} ,
\nonumber
\end{eqnarray*}
with (c) holding for some $C >0$ and $\lambda$-almost all $\eta\in
\Gamma_0$. Then there exists a unique $\mu \in \mathcal{P}_{\rm
exp}(\Gamma)$ for which $k$ is the correlation function.
\end{proposition}

\subsection{The model}

The model which we study is specified by the operator given in
(\ref{1L}). The jump kernel $a$ is supposed to satisfy
\begin{equation} \label{As1}
a(x,y) = a(y,x)\geq 0, \qquad \sup_{y\in \mathds{R}^d}
\int_{\mathds{R}^d} a(x,y) dx = 1.
\end{equation}
Regarding the quantities in (\ref{M1A}) we assume
\begin{equation}
\label{M6A}
 \sup_{x,y \in \mathds{R}^d} \int_{\mathds{R}^d} b(x,y|z) d z
=:\langle b \rangle < \infty , \quad \sup_{x,y, z \in \mathds{R}^d}
b(x,y|z) =: \bar{b} < \infty,
\end{equation}
Moreover, let us define
\begin{eqnarray}
  \label{As2}
  \phi_{+} (x,y) & = & \int_{\mathds{R}^d} a(z,x) b(z,x|y) d z,
 \\[.2cm]
\phi_{-} (x,y) & = & \int_{\mathds{R}^d} a(x,z) b(x,z|y) d z.
\nonumber
\end{eqnarray}
By (\ref{As1}) and (\ref{M6A}) we have that
\begin{equation}
  \label{As2a}
\phi_{\pm}(x,y) \leq \bar{b}, \qquad \ \ {\rm for} \ {\rm all} \ x,y
\in \mathds{R}^d.
\end{equation}
\begin{remark}
  \label{Asrk1}
The quantities defined in (\ref{As2}) can be given the following
interpretation: $\phi_{+} (x,y)$  is the rate with which the
particle located at $y$ attracts other particles to jump (from
somewhere) to $x$; $\phi_{-} (x,y)$  is the rate with which the
particle located at $y$ forces that  located at $x$ to jump (to
anywhere). In the latter case, the particle at $y$ `pushes out' the
one at $x$. Thus, $\phi_{+} (x,y)$ and $\phi_{-} (x,y)$ can be
called attraction and repulsion rates, respectively.
\end{remark}
Now we set
\begin{equation}
  \label{As3}
  \Phi_{\pm }(\eta) =  \sum_{x\in \eta} \sum_{y\in \eta\setminus x}
\phi_{\pm} (x,y),
\end{equation}
which can be interpreted as the total rates of attraction and
repulsion of the configuration $\eta$, respectively. In addition to
(\ref{As1}) and (\ref{M6A}) we assume that the following holds
\begin{equation}
  \label{As4}
\exists \omega \geq 0 \ \ \forall \eta \in \Gamma_0 \qquad
\Phi_{+}(\eta) \leq \Phi_{-}(\eta ) +  \omega |\eta|.
\end{equation}
Note that, for some $c>0$ and all $\eta\in \Gamma_0$, by (\ref{M6A})
it follows that
\begin{equation}
  \label{BAst5}
\Phi_{-}(\eta ) +  \omega |\eta| \leq c |\eta|^2.
\end{equation}
 According to the condition in (\ref{As4}), the rate of
the jumps from somewhere to points close to the configuration $\eta$
(i.e., those which make $\eta$ denser) is in a sense dominated by
the rate of the jumps to anywhere, which thin it out.

\section{The Result}

\label{S3}

\subsection{The operators}
By means of (\ref{1L}) and (\ref{1g}) we calculate $L^\Delta$ and
present it in the form
\begin{equation}
  \label{O1}
 L^\Delta = A^\Delta + B^\Delta + C^\Delta + D^\Delta,
\end{equation}
with the entries
\begin{gather}
  \label{O2}
(A^\Delta k)(\eta)  =  \sum_{y\in \eta}\int_{\mathds{R}^d} a (x,y)
\left( 1 +
 \sum_{z\in \eta\setminus y} b(x,y|z)\right)
k(\eta\setminus y \cup x) dx, \qquad \\[.2cm]
(B^\Delta k)(\eta)  =  - \Psi(\eta) k(\eta), \nonumber
\end{gather}
where
\begin{equation}
  \label{O3}
  \Psi (\eta) = \sum_{x\in \eta} \int_{\mathds{R}^d} a(x,y) d y +
  \Phi_{-} (\eta).
\end{equation}
Furthermore,
\begin{gather}
  \label{O4}
 (C^\Delta k)(\eta) = \sum_{y\in \eta}
\int_{\mathds{R}^d}\int_{\mathds{R}^d}
a(x,y) b(x,y|z) k(\eta\setminus y \cup\{x,z\}) dx dz ,\qquad \\[.2cm]
\nonumber (D^\Delta k)(\eta) = -\int_{\mathds{R}^d} \left(\sum_{x\in
\eta} \int_{\mathds{R}^d} a(x,y) b (x,y|z) dy \right) k(\eta\cup z)
d z.
\end{gather}
As mentioned above, instead of directly dealing with the problem in
(\ref{1a}) we pass from $\mu_0$ to the corresponding correlation
function $k_{\mu_0}$ and then consider the problem
\begin{equation}
  \label{16}
\frac{d}{dt} k_t = L^\Delta k_t, \qquad k_t|_{t=0} = k_{\mu_0},
\end{equation}
with $L^\Delta$ given in (\ref{O1}) -- (\ref{O4}). Our aim now is to
place this problem into the corresponding Banach space. By
(\ref{I4}) we conclude that $\mu \in \mathcal{P}_{\rm exp}(\Gamma)$
implies that $k_\mu \in \mathcal{K}_\vartheta$ for some $\vartheta
\in \mathds{R}$. Hence, we assume that $k_{\mu_0}$ lies in some $
\mathcal{K}_{\vartheta_0}$. Then the formulas in (\ref{O1}) --
(\ref{O4}) can be used to define an unbounded operator acting in
some $\mathcal{K}_\vartheta$. Like in \cite{BKKK,KK} we take into
account that, for each $\vartheta'' < \vartheta'$, the space
$\mathcal{K}_{\vartheta''}$ is continuously embedded into
$\mathcal{K}_{\vartheta'}$, see (\ref{w2}), and use the ascending
scale of such spaces. This means that we are going to define
(\ref{16}) in a given $\mathcal{K}_\vartheta$ assuming that
$k_{\mu_0}\in \mathcal{K}_{\vartheta_0}\hookrightarrow
\mathcal{K}_{\vartheta}$.

For $\omega$ as in (\ref{As4}) and $\Psi$ as in (\ref{O3}), we set
\begin{equation}
  \label{O5}
  \Psi_\omega (\eta) = \omega |\eta| + \Psi(\eta).
\end{equation}
In the sequel, along with those as in (\ref{O2}) and (\ref{O3}) we
use the following operators
\begin{gather}
  \label{O6}
 (B^{\Delta,\omega}k )(\eta) = - \Psi_\omega (\eta) k(\eta),
 \\[.2cm]
(C^{\Delta,\omega}k )(\eta) = (C^{\Delta}k )(\eta) + \omega |\eta|
k(\eta). \nonumber
\end{gather}
Then the decomposition (\ref{O1}) can be rewritten
\begin{equation}
  \label{O7}
L^\Delta = A^\Delta + B^{\Delta,\omega} + C^{\Delta,\omega} +
D^{\Delta},
\end{equation}
with $A^\Delta$ and $D^\Delta$ being as above.

For a given $\vartheta \in \mathds{R}$, we define $L^\Delta$ in
$\mathcal{K}_\vartheta$ by means of  the following estimates. For
$k\in \mathcal{K}_\vartheta$, by (\ref{w2}) we have that
\begin{equation}
  \label{O7a}
|k(\eta)| \leq \|k\|_\vartheta e^{\vartheta |\eta|}, \ \ \qquad {\rm
for} \ \lambda-{\rm a.a.} \ \eta \in \Gamma_0.
\end{equation}
By means of the latter estimate and (\ref{As1}), (\ref{M6A}) we
obtain from (\ref{O2}), (\ref{O3}) and (\ref{O6}) that
\begin{gather}
  \label{O8}
\left\vert (C^{\Delta,\omega}k )(\eta) \right\vert \leq \left(
\omega + \langle b \rangle \right) |\eta| \exp\left[ \vartheta
(|\eta|+1)\right] \cdot \|k\|_{\vartheta}, \\[.2cm]
\left\vert (C^{\Delta}k )(\eta) \right\vert \leq \langle b \rangle
|\eta| \exp\left[ \vartheta (|\eta|+1)\right] \cdot
\|k\|_{\vartheta} , \nonumber \\[.2cm] \left\vert (D^{\Delta}k )(\eta)
\right\vert \leq \langle b \rangle |\eta| \exp\left[ \vartheta
(|\eta|+1)\right] \cdot \|k\|_{\vartheta} , \nonumber
\end{gather}
Now we use (\ref{As2a}) and (\ref{As3}) to obtain from (\ref{O2}),
(\ref{O3}), (\ref{O5}) the following
\begin{gather}
  \label{O9}
\left\vert (A^{\Delta}k )(\eta) \right\vert \leq \left( |\eta| +
\bar{b} |\eta|^2\right) e^{\vartheta |\eta|} \cdot \|k\|_{\vartheta}
,\\[.2cm]
\left\vert (B^{\Delta}k )(\eta) \right\vert \leq \left( |\eta| +
\bar{b} |\eta|^2\right) e^{\vartheta |\eta|}
\cdot \|k\|_{\vartheta} .\nonumber \\[.2cm] \left\vert (B^{\Delta,\omega}k
)(\eta) \right\vert \leq \left[ (1+\omega)|\eta| + \bar{b}
|\eta|^2\right] e^{\vartheta |\eta|} \cdot \|k\|_{\vartheta}
.\nonumber
\end{gather}
The estimates in (\ref{O8}) and (\ref{O9}) allow us to define
$(L^\Delta_\vartheta, \mathcal{D}^\Delta_\vartheta)$, where
\begin{equation}
  \label{C14}
\mathcal{D}^\Delta_\vartheta := \{ k \in \mathcal{K}_\vartheta:
|\cdot|^2 k \in \mathcal{K}_\vartheta \},
\end{equation}
\begin{lemma}
  \label{O1lm}
For each $\vartheta'' < \vartheta$, it follows that
$\mathcal{K}_{\vartheta''} \subset \mathcal{D}^\Delta_\vartheta$.
\end{lemma}
\begin{proof}
 By means of (\ref{O7a}) and the inequality $x\exp(-\sigma x) \leq 1/ e \sigma$, $x,
\sigma
>0$, we get from (\ref{O8}) and (\ref{O9}) the following estimate,
\begin{equation*}
 |\eta|^2 \left\vert k(\eta) \right\vert \leq \frac{4}{e^2
(\vartheta - \vartheta'')^2} \|k \|_{\vartheta''} e^{\vartheta
|\eta|},
\end{equation*}
which yields the proof.
\end{proof}
The same estimate and (\ref{O8}), (\ref{O9}) also yield
\begin{equation}
  \label{O11}
\|L^\Delta k \|_\vartheta \leq 2 \left( \frac{1 + \langle b
\rangle}{e(\vartheta - \vartheta'')}+ \frac{4\bar{b}}{e^2 (\vartheta
- \vartheta'')^2}\right) \|k\|_{\vartheta''},
\end{equation}
which allows us to define a bounded linear operator
$L^\Delta_{\vartheta \vartheta''} : \mathcal{K}_{\vartheta''} \to
\mathcal{K}_{\vartheta}$ the norm of which can be estimated by means
of (\ref{O11}). In what follows, we consider two types of operators
defined by the expression in (\ref{O1}) -- (\ref{O4}): (a) unbounded
operators $(L^\Delta_\vartheta, \mathcal{D}^\Delta_\vartheta)$,
$\vartheta\in \mathds{R}$, with domains as in (\ref{C14}) and Lemma
\ref{O1lm}; (b) bounded operators $L^\Delta_{ \vartheta
\vartheta''}$ as just described. These operators are related to each
other in the following way:
\begin{equation}
  \label{24a}
\forall \vartheta'' < \vartheta \ \  \forall k \in
\mathcal{K}_{\vartheta''} \qquad L^\Delta_{\vartheta
 \vartheta''} k = L^\Delta_{\vartheta} k.
\end{equation}

\subsection{The statement}

We assume that the initial state $\mu_0$ is fixed, which determines
$\vartheta_0 \in \mathds{R}$ by the condition that $k_{\mu_0}$ lies
in $\mathcal{K}_{\vartheta_0}$. Since
$\mathcal{K}_{\alpha''}\hookrightarrow \mathcal{K}_{\alpha'}$ for
$\vartheta'' < \vartheta'$, we take the least $\vartheta_0$
satisfying this condition. Then for $\vartheta > \vartheta_0$, we
consider in $\mathcal{K}_{\vartheta}$ the problem, cf. (\ref{16})
and Lemma \ref{O1lm},
\begin{equation}
  \label{T1}
\frac{d}{dt} k_t = L^\Delta_\vartheta k_t, \qquad k_t|_{t=0} =
k_{\mu_0} \in \mathcal{K}_{\vartheta_0}.
\end{equation}
\begin{definition}
  \label{S1df}
By a (classical) solution of (\ref{T1}) on a time interval, $[0,T)$,
$T\leq +\infty$, we mean a continuous map $[0,T)\ni t \mapsto k_t
\in \mathcal{D}^\Delta_\vartheta$ such that the map $[0,T)\ni t
\mapsto d k_t / dt\in \mathcal{K}_\vartheta$ is also continuous and
both equalities in (\ref{T1}) are satisfied.
\end{definition}
For $\omega\geq 0$ as in (\ref{O5}) and (\ref{O7}), we set, cf.
(\ref{M6A}),
\begin{equation}
  \label{T2}
T(\vartheta, \vartheta_0) = \frac{(\vartheta-
\vartheta_0)e^{-\vartheta}}{\omega + 2 \langle b \rangle},
\end{equation}
where $\vartheta$ and  $\vartheta_0$ are as in (\ref{T1}).
\begin{theorem}
  \label{1tm}
Let the conditions in (\ref{As1}) -- (\ref{As4}) be satisfied. Then,
for each $\vartheta > \vartheta_0$, the problem in (\ref{T1}) has a
unique solution on the time interval $[0, T(\vartheta,
\vartheta_0))$.
\end{theorem}
The proof of this statement will be done in Section \ref{S4} below.
Let us now analyze how to choose $\vartheta$ in an optimal way.
Since the length $T(\vartheta, \vartheta_0)$ of the time interval in
Theorem \ref{1tm} depends on the choice of $\vartheta$, we take
$\vartheta= \vartheta_*$ defined by the condition
\[
T(\vartheta_*, \vartheta_0)= \max_{\vartheta > \vartheta_0}
T(\vartheta, \vartheta_0),
\]
which by (\ref{T2}) yields $\vartheta_* = 1 + \vartheta_0$. Hence,
the maximum length of the time interval is
\begin{equation*}
\tau (\vartheta_0) = T(1+\vartheta_0 , \vartheta_0) =
\frac{e^{-\vartheta_0}}{e(\omega + 2 \langle b \rangle)}.
\end{equation*}
Note that $\tau (\vartheta_0) \to 0 $ as $\vartheta_0 \to + \infty$.

\subsection{Analyzing the assumptions}
Our main assumption in (\ref{As4}) looks like the stability
condition (with stability constant $\omega \geq0$) for the
interaction potential $\phi= \phi_{-} - \phi_{+}$, see (\ref{As2}),
used in the statistical mechanics of continuum systems of
interacting particles, cf. \cite[Chapter 3]{Ruelle} and also
\cite[Section 3.3]{KK}. A particular case of the kernels is where
they are translation invariant and $b$ has the following form
\begin{equation}
  \label{T4}
b(x,y|z) = \kappa_1 (x-z) + \kappa_2( y-z),
\end{equation}
where $\kappa_i (x) = \kappa_i(-x)\geq 0$ belong to
$L^1(\mathds{R}^d)$. Then
\begin{gather}
  \label{T5}
\phi_{+} (x,y) = (\alpha\ast \kappa_1)(x-y) + \kappa_2(x-y),
\\[.2cm]
\phi_{-} (x,y) = \kappa_1(x-y) + (\alpha\ast \kappa_2)(x-y),
\nonumber
\end{gather}
where $\alpha (x) = a(0, x)$ and $\ast$ denotes the usual
convolution. By (\ref{As1}) and (\ref{As2}) $\alpha$ and both
$\kappa_i$ are integrable. Thus, we can use their transforms
\begin{gather*}
\hat{\alpha}(p) =  \int_{\mathds{R}^d} \alpha (x) \exp\left(i
(p,x)\right)d x, \ \qquad p\in \mathds{R}^d, \\[.2cm]
\hat{\phi}_{\pm} (p) = \int_{\mathds{R}^d} \phi_{\pm} (0,x)
\exp\left(i
(p,x)\right)d x, \\[.2cm]
 \hat{\kappa}_i (p)= \int_{\mathds{R}^d}
\kappa_i (x) \exp\left(i (p,x)\right)d x,  \qquad i=0,1,
\end{gather*}
Note that the left-hand sides here are real. Moreover,
$\hat{\alpha}(p) \leq \hat{\alpha}(0) =1$ and $\hat{\kappa}_i(p)
\leq \hat{\kappa}_i(0)$, $i=1,2$.  Then a sufficient condition for
(\ref{As4}) to be satisfied, see \cite[Section 3.2, Proposition
3.2.7]{Ruelle},  is that the following holds: (a) both
$\phi_{\pm}(0,x)$ are continuous; (b) $\hat{\phi}_{-} (p) \geq
\hat{\phi}_{+} (p)$  for all $p\in \mathds{R}^d$. The latter means
that the potential $\phi = \phi_{-} - \phi_{+}$ is positive definite
(in Bochner's sense). In view of (\ref{T5}), (b) turns into
\begin{equation}
  \label{T6}
\forall p \in \mathds{R}^d \qquad \left(1- \hat{\alpha}(p)\right)
\left( \hat{\kappa}_1(p) - \hat{\kappa}_2(p) \right) \geq 0.
\end{equation}
Thus, a sufficient condition for the latter to hold is
$\hat{\kappa}_1(p) \geq \hat{\kappa}_2(p)$ for all those $p$ for
which $\hat{\alpha}(p) < 1$. An example can be
\begin{equation}
  \label{T6a}
\kappa_i (x) = \frac{1}{(2\pi)^{d/2}\sigma_i} \exp \left( -
\frac{|x|^2}{2\sigma_i^2}  \right), \qquad \sigma_1 < \sigma_2,
\end{equation}
cf. \cite[Proposition 3.8]{KK}.

\subsection{Comments}

First we make some comments on the result of Theorem \ref{1tm}. For
the model specified in (\ref{1L}) with a particular choice of $b$,
which we discuss below, in \cite[Theorem 2 and Proposition 1]{BKKut}
there was constructed a weak solution of the problem in (\ref{16})
on a bounded time interval. Our Theorem \ref{1tm} yields a solution
in the strongest sense -- a classical one -- see Definition
\ref{S1df}, existing, however, also on a bounded time interval. At
the same time, this solution $k_t$ yet may not be the correlation
function of a state. To prove this, one ought to develop a technique
similar to that used in \cite[Section 5]{KK} and based on the use of
Proposition \ref{Gpn}. Noteworthy, the fact, proved in \cite{KK},
that $k_t$ is a correlation function allowed there for continuing to
all $t>0$ the solution primarily obtained on a bounded time
interval. For jump dynamics with repulsion, such continuation was
realized in \cite{BK1,BK2}, also by means of the corresponding
property of $k_t$. However, for the model considered here for such a
continuation to be done proving that the solution $k_t$ is a
correlation function -- and hence is positive in a certain sense --
might not be enough. If this is the case, then the attraction in the
form as in (\ref{1L}) is not consistent with the sub-Poissonicity of
the states and hence essentially changes the dynamics of the model.
We plan to clarify this in our forthcoming work.

Now let us return to discussing the conditions imposed on the model.
As mentioned above, in \cite{BKKut} there was studied the model
specified in (\ref{1L}) with the choices of $b$  (cf. \cite[Eqs. (3)
-- (5)]{BKKut}) which in our notations can be presented as follows:
(i) $b(x,y|z) = \kappa(x-z)$; (ii) $b(x,y|z) = \kappa(y-z)$; (iii)
$b(x,y|z) = [\kappa(x-z) + \kappa(y-z)]/2$. Note that all the three
are particular cases of (\ref{T4}). However, instead od our
condition (\ref{As4}) there was imposed a stronger one,  which in
our context can be written down as
\begin{equation}
  \label{T7}
\forall x \in \mathds{R}^d \qquad \phi_{-} (0,x) \geq \phi_{+}
(0,x).
\end{equation}
In case (i), (\ref{T7}) turns into $\kappa (x) \geq (\alpha \ast
\kappa)(x)$, which is much stronger than $\hat{\kappa} (p) \geq 0$
that follows from (\ref{T6}) in this case. E.g., the latter clearly
holds for the Gaussian kernel $\kappa$, see (\ref{T6a}). In case
(ii), which corresponds to pure attraction, cf. Remark \ref{Asrk1},
(\ref{T7}) turns into  $\kappa (x) \leq (\alpha \ast \kappa)(x)$,
which, in fact, is equivalent to $\kappa =(\alpha \ast \kappa)$. The
latter can be considered as the problem of the existence of strictly
positive fixed points of the corresponding (positive) integral
operator in $L^1(\mathds{R}^d)$. In some cases, this problem has
such points, e.g., if the operator is compact -- by the Krein-Rutman
theorem. In the symmetric case (iii), we have $\phi_{+} = \phi_{-}$,
 and hence (\ref{As4}) trivially holds.

\section{The Proof}

\label{S4}

The main idea of proving Theorem \ref{1tm} is to construct the
family of bounded linear operators $Q_{\vartheta \vartheta_0}(t):
\mathcal{K}_{\vartheta_0} \to \mathcal{K}_\vartheta$ with $t \in [0,
T(\vartheta, \vartheta_0))$ such that the solution of (\ref{T1}) is
obtained in the form
\begin{equation}
  \label{T8}
k_t = Q_{\vartheta \vartheta_0}(t) k_0.
\end{equation}
An important element of this construction is another family of
bounded operators obtained by means of a substochastic semigroup
constructed in the $\mathcal{G}_\vartheta$. We obtain this semigroup
in the next subsection in Proposition \ref{1pn}.

\subsection{An auxiliary semigroup}

For a given $\vartheta \in \mathds{R}$, the formulas in (\ref{O9})
allows one to define the corresponding unbounded operators in
$\mathcal{K}_\vartheta$. The predual space of
$\mathcal{K}_\vartheta$ is $\mathcal{G}_\vartheta$ equipped with the
norm defined in (\ref{w1}). For $A^\Delta$ and $B^{\Delta, \omega}$,
see (\ref{O6}), we introduce $\widehat{A}$ and $\widehat{B}^{
\omega}$ by setting, cf. (\ref{9f}),
\[
\langle \! \langle \widehat{A} G, k \rangle \! \rangle = \langle \!
\langle G,  B^{\Delta, \omega} k \rangle \! \rangle , \qquad \langle
\! \langle \widehat{B}^{ \omega} G, k \rangle \! \rangle = \langle
\! \langle G, B^{\Delta, \omega} k \rangle \! \rangle.
\]
This yields
\begin{gather}
  \label{MB5}
(\widehat{A} G) (\eta)  =  \sum_{x\in \eta}\int_{\mathds{R}^d} a
(x,y) \left( 1 +
 \sum_{z\in \eta\setminus x} b(x,y|z)\right)
G(\eta\setminus x \cup y) dy ,\qquad \\[.2cm]
(\widehat{B}^{ \omega} G)(\eta) = - \Psi_\omega (\eta) G(\eta).
\nonumber
\end{gather}
Now for a given $\vartheta$, we set, cf. (\ref{C14})
\begin{equation}
  \label{CB14}
\widehat{\mathcal{D}}^\omega_\vartheta := \{ G \in
\mathcal{G}_\vartheta: \Psi_\omega G \in \mathcal{G}_\vartheta \}.
\end{equation}
Clearly, the multiplication operator $\widehat{B}^{
\omega}_\vartheta: \widehat{\mathcal{D}}^\omega_\vartheta\subset
\mathcal{G}_\vartheta \to \mathcal{G}_\vartheta$ defined in the
second line of (\ref{MB5}) is closed. Moreover, it generates a
$C_0$-semigroup $\{S^{(0)}_\vartheta (t)\}_{t\geq 0}$ of bounded
multiplication operators $(S^{(0)}_\vartheta (t) G)(\eta) = \exp( -
t \Psi_\omega (\eta) ) G(\eta)$. Note that each operator is a
positive contraction, i.e., $S^{(0)}_\vartheta(t)$ maps
$$\mathcal{G}^{+}_\vartheta := \{ G \in \mathcal{G}_\vartheta:
G(\eta ) \geq 0 , \ \lambda-{\rm a.a.} \ \eta \in \Gamma_0\}$$ into
itself and $|S^{(0)}_\vartheta (t) G |_\vartheta \leq
|G|_\vartheta$, see (\ref{w1}). That is, $\{S^{(0)}_\vartheta
(t)\}_{t\geq 0}$ is a \emph{substochastic} semigroup.

For $G \in \widehat{\mathcal{D}}^{\omega, +}_\vartheta :=
\widehat{\mathcal{D}}^\omega_\vartheta \cap
\mathcal{G}^{+}_\vartheta $, by  (\ref{As4}) and (\ref{MB5}) we have
\begin{eqnarray}
  \label{CB7}
|\widehat{A} G|_{\vartheta} & = & \int_{\Gamma_0} \left( \sum_{y\in
\eta} \int_{\mathds{R}^d} a(x,y) d x + \Phi_{+} (\eta)\right)
 G(\eta)
e^{\vartheta|\eta|} \lambda (d\eta) \\[.2cm]
& \leq & \int_{\Gamma_0} \left( \sum_{x\in \eta} \int_{\mathds{R}^d}
a(x,y) d y + \omega |\eta| + \Phi_{-} (\eta)\right) G(\eta)
e^{\vartheta|\eta|} \lambda (d\eta) \nonumber \\[.2cm]
& = & - \int_{\Gamma_0} \left( \widehat{B}^{\omega} G\right)(\eta)
e^{\vartheta|\eta|} \lambda (d\eta). \nonumber
\end{eqnarray}
Likewise, for $G \in \widehat{\mathcal{D}}^{\omega}_\vartheta$ we
get
\begin{equation}
  \label{T9}
|\widehat{A} G|_{\vartheta} \leq \int_{\Gamma_0} \Phi_\omega (\eta)
\left\vert G (\eta) \right\vert e^{\vartheta|\eta|} \lambda (d\eta),
\end{equation}
which means that $\widehat{A}$ can be defined on
$\widehat{\mathcal{D}}^{\omega}_\vartheta$, see (\ref{CB14}).
\begin{lemma}
  \label{T1lm}
The closure $\widehat{T}_\vartheta$  of $(\widehat{A} +
\widehat{B}^\omega, \widehat{\mathcal{D}}^{\omega}_\vartheta)$ in
$\mathcal{G}_\vartheta$ is the generator of a substochastic
semigroup.
\end{lemma}
\begin{proof}
We use the Thieme-Voigt perturbation technique \cite{TV}, see also
\cite[Section 3]{K}. For each $G \in \mathcal{G}_\vartheta^{+}$, we
have that
\begin{equation*}
|G|_\vartheta = \varphi_\vartheta(G) := \int_{\Gamma_0} G(\eta)
e^{\vartheta|\eta|} \lambda (d\eta).
\end{equation*}
Clearly, $\varphi_\vartheta$ is a positive linear functional on
$\mathcal{G}_\vartheta$, and thus the norm defined in (\ref{w1}) is
additive on the cone $\mathcal{G}_\vartheta^{+}$. For $\vartheta' >
\vartheta$, by (\ref{w1}) $\mathcal{G}_{\vartheta'}$ is densely and
continuously embedded into $\mathcal{G}_\vartheta$. Moreover, the
mentioned above semigroup $\{S^{(0)}_\vartheta (t)\}_{t\geq 0}$ has
the property $S^{(0)}_\vartheta (t): \mathcal{G}_{\vartheta'} \to
\mathcal{G}_{\vartheta'}$, $t\geq 0$, and the restrictions
$S^{(0)}_\vartheta (t)|_{\mathcal{G}_{\vartheta'}}$ constitute a
$C_0$-semigroup, which is just $\{S^{(0)}_{\vartheta'} (t)\}_{t\geq
0}$ generated by $(\widehat{B}^\omega_{\vartheta'},
\widehat{\mathcal{D}}^{\omega}_{\vartheta'})$. By (\ref{T9}) we have
that $\widehat{A} : \widehat{\mathcal{D}}^{\omega}_{\vartheta'} \to
\mathcal{G}_{\vartheta'}$. Then according to \cite[Theorem 2.7]{TV},
see also \cite[Proposition 3.2]{K}, the proof will be done if we
show that, for some $\vartheta' > \vartheta$, the following holds
\begin{equation}
  \label{T11}
\forall  G\in \widehat{D}^{\omega,+}_{\vartheta'}\qquad
\varphi_{\vartheta'} ((\widehat{A} + \widehat{B}^\omega)G) \leq
\varphi_{\vartheta'} (G) - \varepsilon \varphi_\vartheta
(\Psi_\omega G)
\end{equation}
with some $\varepsilon>0$. Since (\ref{CB7}) holds for each
$\vartheta\in \mathds{R}$, we have that $$\varphi_{\vartheta'}
((\widehat{A} + \widehat{B}^\omega)G) \leq 0.$$ Then (\ref{T11})
turns into $\varphi_\vartheta (\Psi_\omega G) \leq (1/\varepsilon)
\varphi_{\vartheta'} (G)$. By (\ref{BAst5}) the latter holds for
each $\vartheta'>\vartheta$ and the correspondingly small
$\varepsilon$.
\end{proof}
Let $S_\vartheta:=\{S_\vartheta(t)\}_{t\geq 0}$ be the semigroup as
in Lemma \ref{T1lm}. The semigroup which we need is the sun-dual to
$S_\vartheta$. It is introduced as follows. Let $T^*_\vartheta$ be
the adjoint to the generator of $S_\vartheta$ with domain  ${\rm
Dom} (T^*_\vartheta) \subset\mathcal{K}_\vartheta$ defined in a
standard way. That is,
\[
{\rm Dom} (T^*_\vartheta) = \{ k \in \mathcal{K}_\vartheta: \exists
q \in \mathcal{K}_\vartheta \ \forall G \in
\widehat{\mathcal{D}}^\omega_\vartheta \ \langle \! \langle
\widehat{T}_\vartheta G, k \rangle \!\rangle = \langle \! \langle G,
q\rangle \!\rangle \}.
\]
For each $k\in {\rm Dom} (T^*_\vartheta)$, we have
that
\begin{equation}
  \label{T12}
(T^*_\vartheta k) (\eta) = (A^\Delta k)(\eta) + (B^{\Delta,\omega}
k)(\eta),
\end{equation}
see (\ref{O2}) and (\ref{O6}). By (\ref{O9}) we then get that
$\mathcal{K}_{\vartheta''} \subset {\rm Dom}(T^*_\vartheta)$ for
each $\vartheta'' < \vartheta$. Let $\mathcal{Q}_\vartheta$ stand
for the closure of ${\rm Dom}(T^*_\vartheta)$ in
$\mathcal{K}_\vartheta$. Then
\begin{equation}
 \label{z33}
\mathcal{Q}_\vartheta:= \overline{{\rm Dom}(T^*_\vartheta)}\supset
{\rm Dom}(T^*_\vartheta) \supset \mathcal{K}_{\vartheta''}.
\end{equation}
For each $t\geq 0$, the adjoint $(S_\vartheta(t))^*$ of
$S_\vartheta(t)$ is a bounded operator in $\mathcal{K}_\vartheta$.
However, the semigroup $\{(S_\vartheta(t))^*\}_{t\geq 0}$ is not
strongly continuous. For $t>0$, let $ S^{\odot}_\vartheta(t) $
denote the restriction of $(S_\vartheta(t))^*$ to
$\mathcal{Q}_\vartheta$. Since $S_\vartheta$ is the semigroup of
contractions, for $k\in \mathcal{Q}_\vartheta$ and all $t\geq 0$, we
have that
\begin{equation}
 \label{ACa}
\|S^{\odot}_\vartheta(t)k\|_\vartheta = \|S^{*}(t)k\|_\vartheta \leq
\|k\|_\vartheta.
\end{equation}
\begin{proposition}
 \label{1pn}
For every $\vartheta'' <\vartheta$ and any $k\in
\mathcal{K}_{\vartheta''}$, the map
\begin{equation*}
[0, +\infty) \ni t \mapsto S^{\odot}_\vartheta(t) k \in
\mathcal{K}_\vartheta
\end{equation*}
is continuous.
\end{proposition}
\begin{proof}
By \cite[Theorem 10.4, page 39]{Pazy}, the collection
$S^{\odot}_\vartheta:=\{S^{\odot}_\vartheta(t)\}_{t\geq 0}$
constitutes a $C_0$-semigroup on $\mathcal{Q}_\vartheta$ the
generator of which, $T^{\odot}_\vartheta$, is the part of
$T^*_\vartheta$ in $\mathcal{Q}_\vartheta$. That is,
$T^{\odot}_\vartheta$ is the restriction of $T^*_\vartheta$ to the
set
\begin{equation*}
{\rm Dom} (T^{\odot}_\vartheta):= \{ k\in {\rm Dom}(T^*_\vartheta):
T^*_\vartheta k \in \mathcal{Q}_\vartheta\},
\end{equation*}
cf. \cite[Definition 10.3, page 39]{Pazy}. The continuity in
question follows by the $C_0$-property of the semigroup
$\{S^{\odot}_\vartheta(t)\}_{t\geq 0}$ and (\ref{z33}).
\end{proof}
 By (\ref{O9}) it follows that ${\rm Dom} (T^{\odot}_{\vartheta'}) \supset
\mathcal{K}_{\vartheta''}$, holding  for each $\vartheta'' <
\vartheta'$. Hence, see  \cite[Theorem 2.4, page 4]{Pazy},
\begin{equation*}
S^{\odot}_{\vartheta'} (t) k \in {\rm Dom} (T^{\odot}_{\vartheta'}),
\end{equation*}
and
\begin{equation}
  \label{AEmar}
\frac{d}{dt} S^{\odot}_{\vartheta'} (t) k = A^{\odot}_{\vartheta'}
S^{\odot}_{\vartheta'} (t) k,
\end{equation}
which holds for all $\vartheta' \in (\vartheta'', \vartheta]$ and
$k\in \mathcal{K}_{\vartheta''}$.

\subsection{Getting the solutions}
Here we construct the family of the operators $Q_{\vartheta
\vartheta_0}(t)$ which appear in (\ref{T8}). To this end we use the
semigroup as in Proposition \ref{1pn}.

For $\vartheta'' < \vartheta$, let
$\mathcal{L}(\mathcal{K}_{\vartheta''}, \mathcal{K}_{\vartheta})$
denote the Banach space of bounded linear operators acting from
$\mathcal{K}_{\vartheta''}$ to $\mathcal{K}_{\vartheta}$. By means
of the estimates in (\ref{O9}) one can introduce
$A^\Delta_{\vartheta \vartheta''}$ and $B^{\Delta,\omega}_{\vartheta
\vartheta''}$, both in $\mathcal{L}(\mathcal{K}_{\vartheta''},
\mathcal{K}_{\vartheta})$. Then, cf. (\ref{24a}) and (\ref{T12}),
\begin{equation}
  \label{T13}
\forall k \in \mathcal{K}_{\vartheta''} \qquad T^{\odot}_\vartheta k
= \left( A^\Delta_{\vartheta \vartheta''} +
B^{\Delta,\omega}_{\vartheta \vartheta''}\right)k.
\end{equation}
Let now $S_{\vartheta \vartheta''}(t)$, $t>0$ be the restriction of
$S^{\odot}_{\vartheta} (t)$ to $\mathcal{K}_{\vartheta''}$. Let also
$S_{\vartheta \vartheta''}(0)$ be the embedding
$\mathcal{K}_{\vartheta''} \hookrightarrow \mathcal{K}_{\vartheta}$.
By (\ref{ACa}) we have that the operator norm of such operators
satisfy
\begin{equation}
  \label{T122}
\forall t\geq 0 \qquad \| S_{\vartheta \vartheta''}(t) \| \leq 1.
\end{equation}
By Proposition \ref{1pn} the map
$$[0,+\infty) \ni t \mapsto S_{\vartheta \vartheta''}(t) \in
\mathcal{L}(\mathcal{K}_{\vartheta''}, \mathcal{K}_{\vartheta})$$ is
continuous, and for each $\vartheta' \in (\vartheta'', \vartheta)$,
the following holds, see (\ref{AEmar}) and (\ref{T13}),
\begin{equation}
  \label{T14}
\frac{d}{dt} S_{\vartheta \vartheta''}(t) = \left(
A^\Delta_{\vartheta \vartheta'} + B^{\Delta,\omega}_{\vartheta
\vartheta'}\right) S_{\vartheta' \vartheta''}(t).
\end{equation}
Now by means of the estimates in (\ref{O8}) one concludes that the
formulas in (\ref{O4}) and (\ref{O6}) can be used to introduce
$C^{\Delta,\omega}_{\vartheta \vartheta''}$ and
$D^{\Delta}_{\vartheta \vartheta''}$, both in
$\mathcal{L}(\mathcal{K}_{\vartheta''}, \mathcal{K}_{\vartheta})$.
Their operator norms satisfy, cf. (\ref{O11}),
\begin{equation}
  \label{T15}
\|C^{\Delta,\omega}_{\vartheta \vartheta''}\| \leq \frac{(\omega
+\langle b \rangle)e^\vartheta}{e(\vartheta - \vartheta'')}, \qquad
\|D^{\Delta}_{\vartheta \vartheta''}\| \leq \frac{\langle b \rangle
e^\vartheta}{e(\vartheta - \vartheta'')}.
\end{equation}
Let $\vartheta_0$ be as in (\ref{T1}). Take $\vartheta>\vartheta_0$
and then define
\begin{equation*}
\mathcal{A}(\vartheta, \vartheta_0) = \{ (\vartheta_1 , \vartheta_2,
t): \vartheta_0 \leq \vartheta_1 < \vartheta_2 \leq \vartheta, \ \
0\leq t < T(\vartheta_2, \vartheta_1) \},
\end{equation*}
where $T(\vartheta_2, \vartheta_1)$ is as in (\ref{T2}).
\begin{lemma}
  \label{T2lm}
For any $(\vartheta_1 , \vartheta_2, t) \in \mathcal{A}(\vartheta,
\vartheta_0)$, there exists $Q_{\vartheta_2 \vartheta_1}(t) \in
\mathcal{L}(\mathcal{K}_{\vartheta_1}, \mathcal{K}_{\vartheta_2})$
such that the family $\{Q_{\vartheta_2 \vartheta_1}(t): (\vartheta_1
, \vartheta_2, t) \in \mathcal{A}(\vartheta, \vartheta_0) \}$ has
the following properties:
\begin{itemize}
    \item[{\it (i)}] the map $[0, T(\vartheta_2, \vartheta_1)) \ni t \mapsto  Q_{\vartheta_2
  \vartheta_1}(t) \in \mathcal{L}(\mathcal{K}_{\vartheta_1}, \mathcal{K}_{\vartheta_2})$ is continuous;
\item[{\it (ii)}] the operator norm of $Q_{\vartheta_2
  \vartheta_1}(t)$ satisfies
\begin{equation}
  \label{Herf}
  \|Q_{\vartheta_2
  \vartheta_1}(t)\| \leq \frac{T(\vartheta_2, \vartheta_1)}{T(\vartheta_2, \vartheta_1)
  -t};
\end{equation}
\item[{\it (iii)}] for each $\vartheta_3 \in (\vartheta_1 , \vartheta_2)$ and $t< T(\vartheta_3, \vartheta_1)$,
the following holds
\end{itemize}
\begin{equation}
  \label{bed7}
\frac{d}{dt} Q_{\vartheta_2
  \vartheta_1}(t) = L^\Delta_{\vartheta_2\vartheta_3} Q_{\vartheta_3
  \vartheta_1}(t),
\end{equation}
where $L^\Delta_{\vartheta_2\vartheta_3}$ is as in (\ref{24a}).
\end{lemma}
The proof of this statement employs the following construction. For
$l\in \mathbb{N}$ and $t>0$, we set
\begin{equation*}
\mathcal{T}_l :=\{ (t,t_1, \dots, t_l) : 0\leq t_l \leq \cdots \leq
t_1 \leq t\},
\end{equation*}
fix some $\theta \in (\vartheta_1 , \vartheta_2]$, and then take
$\delta < \theta - \vartheta_1$. Next we divide the interval
$[\vartheta_1,\theta]$ into subintervals with endpoints
$\vartheta^s$, $s=0, \dots, 2l +1$, as follows. Set $\vartheta^0
=\vartheta_1$, $\vartheta^{2l+1 } = \theta$, and
\begin{eqnarray}
\label{z40} \vartheta^{2s} & = & \vartheta_1 +\frac{s}{l+1}\delta +
s \epsilon, \qquad
\epsilon = (\theta - \vartheta_1 - \delta)/l, \\[.2cm]
\vartheta^{2s+1} & = & \vartheta_1+ \frac{s+1}{l+1}\delta + s
\epsilon, \qquad s= 0,1, \dots, l. \nonumber
\end{eqnarray}
Then for $(t,t_1, \dots, t_l) \in \mathcal{T}_l$, define
\begin{gather}
  \label{Iwo}
\Pi_{\theta \vartheta_1}^{(l)}(t,t_1, \dots , t_l) = S_{\theta
\vartheta^{2l}}(t-t_1)
\left(C^{\Delta,\omega}_{\vartheta^{2l}\vartheta^{2l-1} }+
D^{\Delta}_{\vartheta^{2l}\vartheta^{2l-1} } \right) \times \\[.2cm]
\times \cdots \times  S_{\vartheta^{2s+1}
\vartheta^{2s}}(t_{l-s}-t_{l-s+1})
\left(C^{\Delta,\omega}_{\vartheta^{2s}\vartheta^{2s-1} }+
D^{\Delta}_{\vartheta^{2s}\vartheta^{2s-1} } \right) \nonumber
\times \\[.2cm]
\times \cdots \times S_{\vartheta^{3} \vartheta^{2}}(t_{l-1}-t_{l})
\left(C^{\Delta,\omega}_{\vartheta^{2}\vartheta^{1} }+
D^{\Delta}_{\vartheta^{2}\vartheta^{21} } \right)S_{\vartheta^{1}
\vartheta_1}(t_{l}).  \nonumber
\end{gather}
\begin{proposition}
  \label{N1pn}
For each $l\in \mathbb{N}$, the operators defined in (\ref{Iwo})
have the following properties:
\begin{itemize}
\item[{\it (i)}] for each $(t, t_1 , \dots , t_l) \in \mathcal{T}_l$,
$\Pi_{\theta\vartheta_1}^{(l)}(t,t_1, \dots , t_l)\in
\mathcal{L}(\mathcal{K}_{\vartheta_1}, \mathcal{K}_{\theta})$, and
the map
\[
\mathcal{T}_l \ni (t,t_1, \dots, t_l) \mapsto
\Pi_{\theta\vartheta_1}^{(l)} (t,t_1, \dots, t_l) \in
\mathcal{L}(\mathcal{K}_{\vartheta_1}, \mathcal{K}_{\theta})
\]
is continuous;
\item[{\it (ii)}] for fixed $t_1,t_2, \dots , t_l$, and
each $\varepsilon>0$, the map
\[
(t_1, t_1 + \varepsilon) \ni t \mapsto
\Pi_{\theta\vartheta_1}^{(l)}(t,t_1, \dots , t_l) \in
\mathcal{L}(\mathcal{K}_{\vartheta_1}, \mathcal{K}_{\vartheta_2})
\]
is continuously differentiable and for each $\vartheta' \in
(\vartheta_1, \theta)$ the following holds
\end{itemize}
\begin{equation}
  \label{N1p}
\frac{d}{dt} \Pi_{\theta\vartheta_1}^{(l)} (t,t_1, \dots, t_l) =
\left(A^\Delta_{\theta \vartheta'} + B^{\Delta, \omega}_{\theta
\vartheta'}\right) \Pi_{\vartheta'\vartheta_1}^{(l)} (t,t_1, \dots,
t_l).
\end{equation}
\end{proposition}
\begin{proof}
The first part of claim {\it (i)} follows by (\ref{Iwo}),
(\ref{T122}), and (\ref{T15}). To prove the second part we apply
Proposition \ref{1pn} and (\ref{T14}), which yields (\ref{N1p}).
\end{proof}
 \textit{Proof of Lemma \ref{T2lm}.}
Take any $T< T(\vartheta_2, \vartheta_1)$ and then pick $\theta \in
(\vartheta_1, \vartheta_2]$ and a positive $\delta< \theta -
\vartheta_1$ such that
\[
T < T_\delta := \frac{(\theta - \vartheta_1 -
\delta)e^{-\vartheta_2}}{\omega + 2\langle b \rangle}.
\]
For this $\delta$, take $\Pi^{(l)}_{\theta\vartheta_1}$ as in
(\ref{Iwo}), and then set
\begin{eqnarray}
  \label{KkN}
& & Q_{\theta\vartheta_1}^{(n)} (t) = S_{\theta\vartheta_1}(t) \\[.2cm] & & \quad +
\sum_{l=1}^n \int_0^t\int_0^{t_1} \cdots
\int_0^{t_{l-1}}\Pi_{\theta\vartheta_1}^{(l)}(t,t_1, \dots , t_l) d
t_l \cdots dt_1, \quad n\in \mathbb{N}. \nonumber
\end{eqnarray}
By  (\ref{T122}), (\ref{T15}), and (\ref{z40}) the operator norm of
(\ref{Iwo}) satisfies
\begin{equation}
  \label{Iwv}
\| \Pi_{\theta\vartheta_1}^{(l)} (t,t_1, \dots, t_l;\mathbb{B})\|
\leq \left(\frac{l}{eT_\delta}\right)^l,
\end{equation}
holding for all $l=1, \dots , n$. This yields in (\ref{KkN})
\begin{equation*}
\|Q_{\theta\vartheta_1}^{(n)} (t) -
Q_{\theta\vartheta_1}^{(n-1)}(t)\| \leq \frac{1}{n!} \left(
\frac{n}{e}\right)^n \left( \frac{T}{T_\delta}\right)^n,
\end{equation*}
which implies $$\forall t\in [0,T] \quad Q_{\theta\vartheta_1}^{(n)}
(t) \to Q_{\theta\vartheta_1}(t)\in
\mathcal{L}(\mathcal{K}_{\vartheta_1},\mathcal{K}_{\theta}), \ \
{\rm as} \ \ n\to +\infty.$$ This proves claim (i). The estimate in
(\ref{Herf}) follows from that in (\ref{Iwv}). Now by (\ref{Iwo}),
(\ref{T14}), and (\ref{N1p}) we obtain
\begin{equation*}
\frac{d}{dt} Q_{\vartheta_2\vartheta_1}^{(n)} (t) =
\left(A^\Delta_{\vartheta_2\theta} + B^{\Delta,
\omega}_{\vartheta_2\theta} \right) Q_{\theta\vartheta_1}^{(n)} (t)
+ \left( C^{\Delta, \omega}_{\vartheta_2\theta} +
D^{\Delta}_{\vartheta_2\theta} \right) Q_{\theta\vartheta_1}^{(n-1)}
(t), \quad n\in \mathbb{N}.
\end{equation*}
Then the continuous differentiability of the limit and (\ref{bed7})
follow by standard arguments.
\hfill%
$\square $

Now let $k_t$ be as in (\ref{T8}). Then by (\ref{24a}) and
(\ref{bed7}) we conclude that it has all the properties assumed in
Definition \ref{S1df} and hence solves (\ref{T1}). Then to complete
the proof of Theorem \ref{1tm} we have to show that this is a unique
solution.

\subsection{Proving the uniqueness}

Since the problem in (\ref{T1}) is linear, it is enough to show that
its version with the zero initial condition has only the zero
solution. Let $u_t\in \mathcal{K}_\vartheta$ be a solution of this
version. Take some $\vartheta' >\vartheta$ and then $t>0$ such that
$t< T(\vartheta', \vartheta_0)$. Clearly, $u_t$ solves (\ref{T1})
also in $\mathcal{K}_{\vartheta'}$. Thus, it can be written down in
the following form
\begin{equation}
  \label{fra1}
u_t = \int_0^t S_{\vartheta'\vartheta''} (t-s) \left( C^{\Delta,
\omega}_{\vartheta''\vartheta} + D^{\Delta}_{\vartheta''\vartheta}
\right) u_s d s,
\end{equation}
where $u_t$ on the left-hand side (resp. $u_s$ on the right-hand
side) is considered as an element of $\mathcal{K}_{\vartheta'}$
(resp. $\mathcal{K}_{\vartheta}$) and $\vartheta'' \in (\vartheta,
\vartheta')$. Let us show that for all $t<T(\vartheta,
\vartheta_0)$, $u_t = 0$ as an element of $\mathcal{K}_{\vartheta}$.
In view of the embedding $\mathcal{K}_{\vartheta}\hookrightarrow
\mathcal{K}_{\vartheta'}$, this will follow from the fact that $u_t
= 0$ as an element of $\mathcal{K}_{\vartheta'}$. For a given $n\in
\mathbb{N}$, we set $\epsilon = (\vartheta' - \vartheta)/2n$ and
$\alpha^l = \vartheta + l \epsilon$, $l=0, \dots , 2n$. Then we
repeatedly apply (\ref{fra1}) and obtain
\begin{eqnarray*}
 u_t & = & \int_0^t \int_0^{t_1} \cdots \int_0^{t_{n-1}}
S_{\vartheta'\vartheta^{2n-1}} (t-t_1) \left( C^{\Delta,
\omega}_{\vartheta^{2n-1}\vartheta^{2n-2}} +
D^{\Delta}_{\vartheta^{2n-1}\vartheta^{2n-2}} \right) \times
\\ & \times & \cdots \times S_{\vartheta^2\vartheta^{1}} (t_{n-1}-t_n)
\left( C^{\Delta, \omega}_{\vartheta^{1}\vartheta} +
D^{\Delta}_{\vartheta^{1}\vartheta} \right) u_{t_n}
 d t_n
\cdots d t_1.
\end{eqnarray*}
Like in (\ref{Iwv}), we then get from the latter
\begin{eqnarray*}
\|u_t\|_{\vartheta'} & \leq & \frac{t^n}{n!} \prod_{l=1}^{n}
\|C^{\Delta, \omega}_{\vartheta^{2l-1}\vartheta^{2l-2}} +
D^{\Delta}_{\vartheta^{2l-1}\vartheta^{2l-2}}\| \sup_{s\in [0,t]}
\|u_s\|_{\vartheta} \\[.2cm] &\leq & \frac{1}{n!} \left( \frac{n}{e}\right)^n
\left( \frac{2 t (\omega + 2 \langle b \rangle)
e^{\vartheta'}}{\vartheta' - \vartheta}\right)^n \sup_{s\in [0,t]}
\|v_s\|_{\vartheta}. \nonumber
\end{eqnarray*}
This implies that $u_t =0$ for $t < (\vartheta' - \vartheta)/ 2
(\omega + 2 \langle b \rangle) e^{\vartheta'} $. To prove that $u_t
=0$ for all $t$ of interest one has to repeat the above procedure
appropriate number of times.

\section*{Acknowledgment}
The present research was supported by the European Commission under
the project STREVCOMS PIRSES-2013-612669.


\begin{thebibliography}{ll}



\bibitem{Albev} \newblock S. Albeverio, Yu. G. Kondratiev and M. R\"ockner,
\newblock Analysis and geometry on configuration spaces, \newblock \emph{  J. Func. Anal.} {\bf
154} (1998),  444--500.

\bibitem{BK} \newblock J. Bara{\'n}ska and Yu. Kozitsky, \newblock
Free jump dynamics in continuum, \newblock  \emph{Contemporary
Mathematics} {\bf 653} (2015) 13--23.

\bibitem{BK1} \newblock J. Bara{\'n}ska and Yu. Kozitsky, \newblock
The global evolution of states of a continuum Kawasaki model with
repulsion, Preprint arXiv:1509.02044, 2016.

\bibitem{BK2} \newblock J. Bara{\'n}ska and Yu. Kozitsky, \newblock
A Widom-Rowlinson jump dynamics in the continuum, Preprint
 arXiv:1604.07735, 2016.


\bibitem{BKKK} \newblock Ch. Berns,  Yu. Kondratiev, Yu. Kozitsky and O. Kutoviy,
\newblock Kawasaki dynamics in continuum: Micro- and mesoscopic
descriptions, \newblock \emph{J. Dyn. Diff. Equat.} {\bf 25} (2013),
1027--1056.

\bibitem{BKKut} \newblock Ch. Berns,  Yu. Kondratiev and O. Kutoviy,
\newblock Markov jump dynamics with additive intensities in
continuum: state Evolution and mesoscopic scaling, \newblock \emph{
J. Stat. Phys.} {\bf 161} (2015), 876--901.

\bibitem{KK} \newblock Yu.  Kondratiev and Yu.  Kozitsky, \newblock
The evolution of states in a spatial population model, \newblock
\emph{J. Dyn. Diff. Equat.} (2016); DOI 10.1007/s10884-016-9526-6.


\bibitem{Tobi} \newblock Yu.  Kondratiev and T.  Kuna, \newblock Harmonic analysis on configuration space. I. General
theory, \newblock \emph{ Infin. Dimens. Anal. Quantum Probab. Relat.
Top.} {\bf 5} (2002), 201--233.

\bibitem{K} \newblock Yu. Kozitsky, Dynamics of spatial logistic model:  finite systems, in: J. Banasiak,
A. Bobrowski, M. Lachowicz (Eds.), \newblock Semigroups of Operators
-- Theory and Applications: Beedlewo, Poland, October 2013. Springer
Proceedings in Mathematics \& Statistics 113, Springer 2015, pp.
197--211.

\bibitem{Pazy}\newblock A. Pazy, Semigroups of Linear Operators and Applications to Partial Differential
Equations, \newblock Applied Mathematical Sciences, 44.
Springer-Verlag, New York, 1983.


\bibitem{Ruelle} \newblock D. Ruelle,  Statistical Mechanics: Rigorous
Results, \newblock W. A. Benjamin, Inc., 1969.

\bibitem{TV} H. R. Thieme and J. Voigt, Stochastic semigroups: their construction by perturbation
and approximation, in: M. R. Weber and J. Voigt (Eds.), Positivity
IV -- Theory and Applications, Tech. Univ. Dresden, Dresden, 2006,
pp. 135--146.




\end{thebibliography}
\end{document}